\documentclass[runningheads]{llncs}
\usepackage[utf8]{inputenc}

\usepackage{graphicx}
\usepackage{amsmath}
\usepackage{amssymb}
\usepackage{bbm}
\usepackage{physics}
\usepackage{MnSymbol}

\usepackage[maxbibnames=20,style=numeric]{biblatex}
\usepackage{authblk}

\usepackage{tikz}
\usetikzlibrary{positioning}

\usepackage{algorithm}
\usepackage{algpseudocode}

\newcommand{\ZZ}{\mathbb{Z}}

\newcommand{\NN}{\mathbb{N}}

\newcommand{\pqc}{{\tt PQC}}
\newcommand{\nist}{{\tt NIST}}
\newcommand{\sdpke}{{\tt SDPKE}}
\newcommand{\kem}{{\tt KEM}}
\newcommand{\dss}{{\tt DSS}}

\newcommand{\gemss}{{\tt G{\it e}MSS}}
\newcommand{\rain}{{\tt Rainbow}}
\newcommand{\sike}{{\tt SIKE}}
\newcommand{\sdlp}{{\tt SDLP}}

\newcommand{\spdh}{{\tt SPDH}} 
\newcommand{\gadlp}{{\tt GADLP}}

\newcommand{\ahsp}{{\tt AHSP}}

\newcommand\repeatedtheorem[2]{%
  \expandafter\let\expandafter\repeatedtheoremtmp\csname#1name\endcsname
  \expandafter\def\csname#1name\endcsname{%
    \repeatedtheoremtmp\ \ref{#2}%
    \global\expandafter\let\csname#1name\endcsname\repeatedtheoremtmp
  }
}

\spnewtheorem*{theorem*}{Theorem}{\normalshape\bfseries}{\itshape}

\newcommand{\modd}[1]{[#1]_r}

\newcommand{\ZZr}{\mathbb{Z}_r}
\newcommand{\bigO}[1]{\mathcal{O}(#1)}
\newcommand{\bigOmega}[1]{\Omega(#1)}

\makeatletter

\newcommand*\star@[2]{\mathbin{\vcenter{\hbox{\scalebox{#2}{$\m@th#1\bullet$}}}}}
\makeatother

\title{A Subexponential Quantum Algorithm for the Semidirect Discrete Logarithm Problem}
\titlerunning{A Subexponential Quantum Algorithm for Semidirect Discrete Logarithm}
\author{}
\institute{}

\author{
 Christopher~Battarbee\inst{1,5}
    \and
    Delaram~Kahrobaei\inst{1,2,3,4}
   \and
    Ludovic~Perret\inst{1}
    \and
    Siamak~F.~Shahandashti\inst{5}
}

\institute{
    Sorbonne University, CNRS, LIP6, PolSys, Paris, France
    \and
    Departments of Computer Science and Mathematics, Queens College, City University of New York, USA
    \and
    Initiative for the Theoretical Sciences, Graduate Center, City University of New York, USA
    \and
    Department of Computer Science and Engineering, Tandon School of Engineering, New York University, USA
    \and
    Department of Computer Science, University of York, UK
}

\authorrunning{C. Battarbee, D. Kahrobaei, L. Perret and S. F. Shahandashti}

\begin{document}
\maketitle

\begin{abstract} 
\textit{Group-based} cryptography is a relatively unexplored family in post-quantum cryptography, and the so-called Semidirect Discrete Logarithm Problem (\sdlp{}) is one of its most central problems. However, the complexity of the general case of \sdlp{} and its relationship to more well-known hardness problems, particularly with respect to its security against quantum adversaries, has not been well understood and was a significant open problem for researchers in this area. In this paper we give the first dedicated security analysis of the general case of \sdlp{}. In particular, we provide a connection between \sdlp{} and group actions, a context in which quantum subexponential algorithms are known to apply. We are therefore able to construct a subexponential quantum algorithm for solving \sdlp{}, thereby classifying the complexity of \sdlp{} and its relation to known computational problems. 
\end{abstract}

\section*{Introduction}
The goal of Post-Quantum Cryptography (\pqc{}) is to design cryptosystems which are secure against classical and quantum adversaries. A topic of fundamental research for decades, the status of \pqc{} drastically changed with the \nist{} \pqc{}  standardization process \cite{NISTPQ}.

In July $2022$, after five years and three rounds of selection, \nist{}  selected a first set of \pqc{} standards for Key-Encapsulation Mechanism (\kem) and Digital Signature Scheme  (\dss{}) protocols, based on lattices and hash functions. The  standardization process is still ongoing with a fourth round for \kem{} and a new \nist{} call for post-quantum \dss{} in $2023$. Recent attacks \cite{DBLP:journals/iacr/Beullens22,DBLP:conf/crypto/TaoPD21,DBLP:journals/iacr/BaenaBCPSV21} against round-$3$ multivariate signature schemes, \rain{} \cite{DBLP:journals/iacr/Beullens22} and \gemss{} \cite{casanova2017gemss}, as well as the cryptanalysis of round-$4$ isogeny based \kem{} \sike{} \cite{cryptoeprint:2022/975, cryptoeprint:2022/1026}, emphasise the need to continue the cryptanalysis effort in \pqc{} as well as to increase the diversity in the potential post-quantum hard problems.

A relatively unexplored family of such problems come from \textit{group-based} cryptography (see \cite{kahrobaeinotices,Kahrobaei-BattarbeeBook}).  In particular we are interested in the so-called Semidirect Discrete Logarithm Problem (\sdlp{}), which initially appears in the 2013 work of Habeeb et al.~\cite{habeeb2013public}. Roughly speaking, we generalise the standard notion of group exponentiation by employing products of the form $\phi^{x-1}(g)\cdot\ldots\cdot\phi(g)\cdot g$, where $g$ is an element of a (semi)group, $\phi$ is an endomorphism and $x\in\NN$ is a positive integer. Our task in \sdlp{} is to recover the integer $x$ given the pair $g,\phi$ and the value $\phi^{x-1}(g)\cdot\ldots\cdot\phi(g)\cdot g$. It turns out that products of this form have enough structure to be cryptographically useful, in a sense we will expand upon later - in particular, protocols based on \sdlp{} are plausibly post-quantum, since there is no known reduction of \sdlp{} to a Hidden Subgroup Problem.

By far the most studied such protocol is known as Semidirect Product Key Exchange (\sdpke{}), originally proposed in \cite{habeeb2013public} (note that this is the same work in which \sdlp{} first appears). It is a Diffie--Hellman-like key exchange protocol in which products of the form $\phi^{x-1}(g)\cdot\ldots\cdot\phi(g)\cdot g$ are exchanged between two parties in such a way as to allow both parties to recover the same shared key. Clearly, the security of \sdpke{} and the difficulty of \sdlp{} are heavily related -- in particular, an adversary able to solve \sdlp{} is also able to break \sdpke{}. 

There is therefore motivation to analyse the difficulty of \sdlp{}. However, prior to this work the general-case complexity of \sdlp{} and its relationship to more well-known hardness problems, particularly with respect to its security against quantum adversaries, has not been well understood and was a significant open problem for researchers in this area. In this paper, we provide the first dedicated analysis of \sdlp{}, obtaining two key contributions. First, we demonstrate that a subset of all possible products of the form $\phi^{x-1}(g)\cdot\ldots\cdot\phi(g)\cdot g$ is a set upon which a finite abelian group acts; in other words, that \sdpke{} is, modulo some context-specific technicality, a variant of the group action-based key exchange schemes originally proposed by Couveignes \cite{couveignes2006hard}. In particular, solving \sdlp{} can be translated into a problem with respect to a group action. This surprising connection provides a sharper classification of \sdpke{} than was previously known, and allows us to derive our second contribution, an application of known tools that gives a quantum algorithm for solving \sdlp{}. The algorithm runs in subexponential time $2^{\bigO{\sqrt{\log p}}}$, where $p$ parameterises the size of the group action.  

\subsection*{Related Work}
Examples of concrete proposals for \sdpke{} can be found in \cite{habeeb2013public, grigoriev2019tropical, kahrobaei2016using, rahman2022make, rahmanmobs}; respective cryptanalyses can be found in \cite{myasnikov2015linear, isaac2021closer, brown2021cryptanalysis, monico2021remarks, Battarbee_Kahrobaei_Shahandashti_2022, battarbee2021efficiency}. This body of work proceeds more or less chronologically, in that proposed platforms are a response to cryptanalysis addressing a weakness in an earlier version. 
A more detailed survey of the back-and-forth on this topic can be found in \cite{battarbee2022semidirect}.

Despite the relationship between \sdpke{} and \sdlp{}, none of the works discussed above provide an analysis of \sdlp{}. Indeed, the general direction of research in this area has been either to achieve shared key recovery by exploiting some underlying linearity of a platform (semi)group, or to find examples of (semi)groups with sufficiently lax structure to render these attacks less powerful. In particular, none of the cryptanalyses in this area solve \sdlp{}.

\subsubsection*{Recent algorithms for \sdlp{}.}\label{sdlp-algos} Since the preparation of this manuscript various works either related to it or advancing upon it have appeared. We detail them briefly below.

The authors of \cite{battarbee2023spdh} derive a digital signature scheme whose security is based on the difficulty of \sdlp{} in a certain $p$-group. Using similar techniques to that of this paper, they show that the resulting signature scheme admits quantum subexponential time forgeries.

Two papers dedicated to the analysis of \sdlp{}, meanwhile, have appeared rather recently. The first is the work of \cite{imran2023efficient}, whose basic idea is use a clever recursion tool in the quotients of the underlying group to reduce some instance of \sdlp{} to an instance of \sdlp{} in an elementary abelian group. Moreover, an existing quantum algorithm due to \cite{ivanyos2001efficient} allows one to compute the necessary series of quotients; one is therefore done after showing that this abelian instance of \sdlp{} is solved simply by solving a classical discrete logarithm problem in quantum polynomial time. The second of these papers is \cite{mendelsohn2023small}, which is more specifically tailored to the platform group proposed in \cite{battarbee2023spdh} -- again, however, it exploits the decomposition of that platform group into abelian constituents. Interestingly, \cite{mendelsohn2023small} makes progress towards demonstrating the quantum equivalence of \sdlp{} and the Computational Diffie-Hellman-like problem underpinning the security of \sdpke{}, highlighting the potential need for future vigilance against \sdpke{} attacks when considering the difficulty of \sdlp{}.

All three of these works make reference to a pre-print of this paper, and indeed make use of insights and techniques first defined therein. Nevertheless, the chief novelty of this work in comparison with its peers is the generality of the method. As one might expect from the requirement of the groups in \cite{imran2023efficient} to have some abelian quotient, the principal class of group admitting quantum polynomial algorithms for \sdlp{} are the solvable groups (though the work does present a couple of other examples related to matrix groups over finite fields). Our algorithm is not restricted to this class of groups, and in fact is able to solve \sdlp{} when the candidate group is replaced by a semigroup. On the other hand, in these solvable groups the method of \cite{imran2023efficient} runs in quantum polynomial time - i.e. a significant speed-up of our general method. We therefore remark that in order to understand where \sdlp{} fits in the general landscape of potentially post-quantum hard computational problems, it is important to consider both the current paper and the works discussed in this section.  

\subsubsection*{Related techniques.} Ideas much closer to the spirit of our work appear in papers that, at first glance, appear unrelated to \sdpke{} and \sdlp{}. Our results are achieved in part by careful synthesis of the techniques in the two papers of Childs et al. \cite{childs2014constructing, childs2014quantum}: since the set of all products of the form $\phi^{x-1}(g)\cdot\ldots\cdot\phi(g)\cdot g$ admits some similarity to that of a monogenic semigroup, we can adapt some ideas from a quantum algorithm in \cite{childs2014quantum} that solves the Semigroup Discrete Logarithm Problem. However, in our setting we are lacking some key structure that allows the direct application of \cite{childs2014quantum}. The full algorithm is constructed by adapting ideas in \cite{childs2014constructing}, allowing us to show the important quantum algorithms of Kuperberg \cite{kuperberg2005subexponential} and Regev \cite{regev2004subexponential} can be used to solve \sdlp{}.


\subsection*{Organisation of the Paper and Main Results}
The construction of the algorithm claimed in the title is, from a high-level perspective, achieved by two reduction proofs followed by an application of known algorithms. With this in mind, we make the following contributions. 

\subsubsection{Section~\ref{prelims}.}
We start with preliminaries in which the necessary background is reviewed. In particular we give a brief discussion of the relevant algebraic objects, a short note on quantum computation, and a full description of \sdpke{} - including a discussion of appropriate asymptotic assumptions for the growth of the semigroups we work with, derived from currently proposed platforms. In particular we derive a notion of a parametrised family of semigroups called an `easy' family of semigroups. The section finishes with a kind of glossary of computational problems. 

\subsubsection{Section~\ref{group-action-section}.} 
It will be immediately convenient to write $s(g,\phi,x)$ for $\phi^{x-1}(g)\cdot...\cdot g$, not only for clarity of notation but to aid the crucial shift in perspective offered in this paper. Armed with this notation our first task is to study the set of possible exponents\footnote{If $(g,\phi)$ is fixed we can think of $s(g,\phi,\cdot)$ as a function from $\NN$ into $G$ that generalises the usual group exponentiation map, in the sense that the usual exponentiation map is recovered if $\phi$ is the identity morphism. Within this intuition we shall feel comfortable referring to the values of $s$ relative to some fixed $(g,\phi)$ as `exponents'.} $\{s(g,\phi,i):i\in\mathbb{N}\}$. In Theorem~\ref{exponent-structure} we deduce, borrowing from some standard ideas in semigroup theory, that this set is finite and has the form $\{g,...,s(g,\phi,n),...,s(g,\phi,n+r-1)\}$ where the integers $n,r$ are a function of the choice of $(g,\phi)$ (and may, when desirable, be written $r_{g,\phi},n_{g,\phi}$ to highlight this point). The main result of this section is Theorem~\ref{free-transitive-action}: an abelian group acts freely and transitively on the set $\{s(g,\phi,n),...,s(g,\phi,n+r-1)\}$. This set is called the cycle of $g,\phi$ and is denoted by the calligraphic letter $\mathcal{C}$. In particular, we show the following:

\begin{theorem*}
Fix $(g,\phi)\in G\times End(G)$ and let $n,r$ be the index and period corresponding to $g,\phi$. Moreover, let $\mathcal{C}$ be the corresponding cycle of size $r$. The abelian group $\mathbb{Z}_r$ acts freely and transitively on $\mathcal{C}$.
\end{theorem*}

where $\ZZr$ is the usual notion of a group of residues modulo $r$. 

\subsubsection{Section~\ref{group-action-discrete-log}.}
If the set of exponents with respect to $(g,\phi)$ consisted entirely of the cycle we would immediately have a reduction to the Group Action Discrete Logarithm Problem (\gadlp{}), also referred to as \textit{Group Action DLog} in \cite{montgomery2022full}, or the \textit{Parallelisation Problem} in \cite{couveignes2006hard}. Roughly speaking, since we know that an abelian group acts on the cycle, in this case the exponent $x$ to be recovered is precisely the group element acting on $g$ to give $s(g,\phi,x)$. This is not, however, generally the case, and to proceed it will be necessary to extract the pair $n,r$ from the base values $g,\phi$. In Theorem~\ref{index-period-recovery} we show that one can achieve this in efficient quantum time by using canonical quantum period-finding methods. We can therefore deduce in Theorem~\ref{sdlp-to-gadlp} that one can solve \sdlp{} efficiently given access to a \gadlp{} oracle; or, if $s(g,\phi,x)$ is not in the cycle of $g,\phi$, by invoking a classical procedure exploiting the knowledge of $n,r$. Indeed, we show the following:

\begin{theorem*}
Let $\{G_p\}_p$ be an easy family of semigroups, and fix $p$. Algorithm~\ref{sd-ga} solves \sdlp{} with respect to a pair $(g,\phi)\in G_p\times End(G_p)$ given access to a \gadlp{} oracle for the group action $(\ZZ_{r_{g,\phi}}, \mathcal{C}_{g,\phi}, \oast)$. The algorithm runs in time $\bigO{(\log p)^4}$, makes at most a single query to the \gadlp{} oracle, and succeeds with probability $\bigOmega{1}$.  
\end{theorem*}

Clearly, many of the requisite notions in this statement have not yet been defined. Roughly speaking, an easy family of semigroups is a family of semigroups parameterised by $p$ such that each of the functions taking $p$ as an argument grows polynomially in $p$.

\subsubsection{Section~\ref{gadlp-algos}.} In order to give a full description and complexity analysis of the algorithm it remains to examine the state of the art for solving \gadlp{}. It is reasonably well-known (see \cite{stolbunov2010constructing}, \cite{childs2014constructing}) that \gadlp{} reduces to the Abelian Hidden Shift Problem if the action is free and transitive, though we provide a context-specific reduction in Theorem~\ref{gadlp-to-ahsp}. There are two popular choices to solve this problem: an algorithm due to Kuperberg \cite{kuperberg2005subexponential} and another due to Regev \cite{regev2004subexponential}, each of which has trade-offs with respect to time and space complexity. Finally, the full algorithm is given in Theorem~\ref{complete-algo}, though we are essentially assembling the components we have developed throughout the rest of the paper. This main result is the following:

\begin{theorem*}
    Let $\{G_p\}_p$ be an easy family of semigroups, and fix $p$. For any pair $(g,\phi)\in G_p\times End(G_p)$, there is a quantum algorithm solving \sdlp{} with respect to $(g,\phi)$ with time and query complexity $2^{\bigO{\sqrt{\log p}}}$. 
\end{theorem*}

which proves the claim of a quantum subexponential algorithm for \sdlp{} given in the title.

\section{Preliminaries}\label{prelims}
\subsection{Notation}
One need only be familiar with the standard $\bigO$ and $\bigOmega$ notations. Various bespoke notations are introduced throughout the course of the paper, but these will be defined in due course. We note also that all our logarithms are base-$2$.

\subsection{Background Mathematics}
We recall a number of group-theoretic notions used throughout this paper. 

Recall that a group for which one is not guaranteed to have inverse with respect to the group operation is a \textit{semigroup}. Writing the operation multipicatively, the semigroups $G$ we are interested in all have an element $1$ such that $1\cdot g=g=g\cdot 1$ for each $g\in G$ - such a semigroup is also called a \textit{monoid}. In this work we insist that we do not have a full group, and so write semigroups without meaning to refer to full groups.

We will deal with both abelian groups and non-abelian semigroups in this paper; that is, for a non-abelian semigroup $G$ one cannot expect that $g\cdot h=h\cdot g$ for all $g,h\in G$. For the sake of clarity we will write abelian groups additively. In this case, the operation $g+h$ commutes, and we require an inverse for every element. Note also that the identity is written as $0$ in this case. 

Consider a function from $G$ to itself, say $\phi$. If $\phi$ preserves multiplication - that is, $\phi(g\cdot h)=\phi(g)\cdot \phi(h)$ for each $g,h\in G$ - we call it an \textit{endomorphism}. Certainly we can compose these functions according to the usual notion, and indeed it is standard that the set of all endomorphisms under function composition defines a semigroup. Since we allow for (and in some cases require) that the endomorphisms are not invertible, we have a semigroup rather than a full group. In particular, every finite semigroup $G$ immediately induces an \textit{endomorphism semigroup}, denoted $End(G)$. 

An important, and in this context frequently invoked source of semigroups come from \textit{matrix algebras}. The set of square matrices of fixed size with entries in some ring $R$ forms an $R$-module, since we can add matrices together and scale each entry of a matrix by some $r\in R$. The necessary distributivity properties are inherited from the properties of $R$. However, unlike in a usual $R$-module, we can also \textit{multiply} elements just by defining multiplication to be the usual notion of matrix multiplication. The resulting matrix algebra is denoted $M_n(R)$, where $n\in\mathbb{N}$ is the fixed size of matrix, and $R$ is the underlying ring. Indeed, consider a matrix algebra under only the multiplication operation. It is again clear that this object is a semigroup; we would have a full group if every matrix was invertible, but of course this is not true. The all-zero matrix, for example, has no multiplicative inverse. A matrix algebra considered only under its multiplication, therefore, is a useful source for concrete examples of semigroups. 

It will be useful for us to build a new semigroup from an existing semigroup. One way of doing this is via a structure called the \textit{holomorph}. Let $G$ be a (semi)group and $End(G)$ its endomorphism semigroup. The holomorph $G\ltimes End(G)$ is the set $G\times End(G)$ equipped with multiplication
\[(g,\phi)\cdot (g',\phi') = (\phi'(g)\cdot g',\phi'\circ\phi)\]
where $\circ$ refers to function composition. In fact, the holomorph is itself a special case of the semidirect product, hence the `semidirect' terminology found throughout this paper.

\subsubsection{Group Actions.}\label{prelims-group-actions}
A key idea for us will be that of a group action, and in particular a commutative group action. Roughly speaking such an object allows one to map elements of a set to each other in a cryptographically useful fashion, but in a less structured manner than in more classical settings. More formally:
\begin{definition}[Commutative Group Action]\label{defn-group-actions}
    Let $G$ be a finite abelian group and $X$ be a finite set. Consider a function from $G\times X\to X$, written by convention as $g\star x$, with the following properties:
    \begin{enumerate}
        \item $1\star x=x$
        \item $(g+ h)\star x = g\star (h\star x)$
    \end{enumerate}
    The tuple $(G,X,\star)$ is a \textit{commutative group action}. If only the identity fixes an arbitrary element of $X$ the action is \textit{free}, and if for any $x,y\in X$ there is a $g$ such that $g\star x=y$ the action is \textit{transitive}.
\end{definition}
The group action defined in this paper is commutative, so we will sometimes just write `group action' to mean a commutative group action. It will remain for us, however, to prove that this action is free and transitive. If the action is indeed free and transitive, it follows that for any $x\in X$, all $y\in X$ are such that there exists a unique $g\in G$ with $g\star x=y$. Borrowing notation from Couveignes \cite{couveignes2006hard}, it will sometimes be convenient for us to write $\delta(y,x)$ to denote this value.

\subsection{Quantum Computation}
In order to present our quantum algorithm for \sdlp{} (Section~\ref{group-action-discrete-log}), the reader needs only be familiar with standard quantum tools, presented for example in \cite{hirvensalo2003quantum}. We give a brief summary of the required notions below.

Recall that $n$ qubits can be represented by the complex vector space $H_{2^n}$, where the basis states are exactly the $n$-fold tensor products of basis states of $H_2$. An ordered system of $n$ qubits is called a \textit{quantum register of length $n$}, and the basis states are sometimes written $\{\ket{i}:0\leq i<2^n\}$ by identifying $i$ with its binary representation. 

Recall also that we can create the uniform superposition efficiently with a Hadarmard gate. For an $l$-qubit register, the computational basis vector $\ket{0}$ is such that the Hadamard gate (written $H_{2^l}$) is such that 
\[H_{2^l}\ket{0}=\frac{1}{\sqrt{2^l}}\sum_{i=0}^{2^l-1}\ket{i}\]
Moreover, this transformation can be carried out efficiently, in time $\bigO{l}$.

\subsection{Semidirect Product Key Exchange}\label{prelimsSDPKE}
We here define in full \sdpke{}. One verifies by induction that holomorph exponentiation takes the form 
\[(g,\phi)^x=(\phi^{x-1}(g)\cdot\ldots\phi(g)\cdot g, \phi^x)\]
where $\phi^x$ denoted the endomorphism $\phi$ composed with itself $x$ times. Note that this operation involves multiplying (semi)group elements, endomorphisms, and applying an endomorphism to a semigroup element. If all these operations are efficient, the holomorph exponentiation is efficient since one can apply standard square-and-multiply techniques.

The central idea of \sdpke{} is to use products of these form as a generalisation of Diffie--Helman Key-Exchange. Suppose $N$ is the number of all possible distinct holomorph exponents - there are finitely many - then the protocol works as follows:
\begin{enumerate}
    \item Suppose Alice and Bob agree on a public (semi)group $G$ and hence the integer $N$, as well as a group element $g$ and endomorphism of $G$, say $\phi$.
    \item Alice picks a secret integer $x$ uniformly at random from $\{1,...,N\}$, and calculates the holomorph exponent $(g,\phi)^x=(A,\phi^x)$. She sends \textbf{only} $A$ to Bob.
    \item Bob similarly calculates $(B,\phi^y)$ corresponding to a random, private integer $y$, and sends only $B$ to Alice.
    \item With her private automorphism $\phi^x$ Alice can now calculate her key as the group element $K_A=\phi^x(B)\cdot A$; Bob similarly calculates his key $K_B=\phi^y(A)\cdot B$.
\end{enumerate}
We have
\begin{align*}
    \phi^x(B)\cdot A &= \phi^x(\phi^{y-1}(g)\cdot\ldots\cdot g)\cdot (\phi^{x-1}(g)\cdot\ldots\cdot g) \\
    &= (\phi^{x+y-1}\cdot\ldots\cdot \phi^x(g)) \cdot (\phi^{x-1}(g)\cdot\ldots\cdot g) \\
    &= (\phi^{x+y-1}\cdot\ldots\cdot\phi^y(g))\cdot(\phi^{y-1}(g)\cdot\ldots\cdot g) \\
    &= \phi^y(A)\cdot B
\end{align*}
so $K:=K_A=K_B$. Note that $A\cdot B\neq K$ as a consequence of our insistence that the endomorphism $\phi$ is non-trivial - note that in the closed form derived in \cite{habeeb2013public} when $\phi$ is a conjugation, we also require $G$ to be non-commutative.

Writing these products in full will quickly become rather cumbersome. We therefore introduce some non-standard notation, which is useful both for convenience of exposition and the required shift in perspective we will introduce in this paper.
\begin{definition}\label{semidirect-product-function}
Let $G$ be a finite, non-commutative (semi)group, $g\in G$, and $\phi\in End(G)$. We define the following function: 
\begin{align*}
    s:G\times End(G)\times \mathbb{N} &\to G \\
     (g,\phi,x) &\mapsto \phi^{x-1}(g)\cdot\ldots\cdot \phi(g)\cdot g   
\end{align*}
\end{definition}

Notice that when $g,\phi$ are fixed - as in the case of the key exchange - the function $s$ is really only taking integer arguments, analogously to the standard notion of group exponentiation. Indeed, a passive adversary observing a round of \sdpke{} has access to the values $s(g,\phi,x)$ and $s(g,\phi,y)$ - in order to recover the shared key $s(g,\phi,x+y)$ one strategy they might adopt is to recover the private integers $x,y$ from $s(g,\phi,x),s(g,\phi,y)$ to allow calculation of $s(g,\phi,x+y)$. In short, the security of \sdpke{} is clearly in some sense related to the Semidirect Discrete Logarithm Problem alluded to in the introduction. We shall have much more to say about this later on.

\subsection{Efficiency Considerations}\label{prelims-efficiency}
The works discussed in the introduction, as well as the contents of the more comprehensive survey \cite{battarbee2022semidirect}, highlight that every extant proposal of a platform for \sdpke{} suggests for use some variety of matrix algebra. In particular, insofar as parameters are recommended, the convention is to fix a matrix size - usually $3$ - and adjust the size of an underlying ring in order to increase security. 

In other words, having defined \sdpke{} above relative to some semigroup $G$ and its endomorphism semigroup $End(G)$, we can think of each such semigroup as one of a family of semigroups $\{G_p\}_p$, where the family $\{G_p\}_p$ is indexed by some set parameterising the underlying algebra (usually the primes). Note that this immediately induces a family of endomorphism semigroups $\{End(G_p)\}_p$, so we can talk about pairs $(g,\phi)$ from the set $G_p\times End(G_p)$ for each $p$. 

Table~\ref{platform-growth} gives examples of platforms over $3\times 3$ matrices, the size of the platform, and the variable that can be considered as the indexing variable\footnote{Note here that $|R|$ is chosen as the parameter for reasons of efficiency of representation.}. 
\begin{table}
    \centering
    \caption{Growth of Proposed Platforms}
    \setlength{\tabcolsep}{0.5em} 
    \renewcommand{\arraystretch}{1.3}
    \begin{tabular}{c|c|c}
      Proposed Platform & Size of Platform & Indexing Variable \\
      \hline
      \hline
      $M_3(G[R])$ & $|R|^{9|G|}$ & $|R|$\\
      Certain classes of $p$-group & Polynomial in prime $p$ & Prime $p$ \\
      $M_3(\mathbb{Z}_p)$ & $p^9$ & $p$ \\
    \end{tabular}
    \label{platform-growth}
\end{table}

In each of these examples we have a family of semigroups indexed by some set $P$ such that each semigroup $G_p$ has size polynomial in $p$. We will give complexity estimates as a function of $p$ - indeed, let us now see how the complexity of executing \sdpke{} grows with $p$. For a semigroup $G$, note that with respect to the holomorph $G\ltimes End(G)$ we have $(g,\phi)^x=(s(g,\phi,x),\phi^x)$ by definition. By standard square-and-multiply techniques it therefore requires $\bigO{\log x}$ applications of the operation in the holomorph to compute $s(g,\phi,x)$.

In order to estimate the complexity of the holomorph operation we need to know the complexity of multiplication in $G$ and that of applying the endomorphism $\phi$. In this direction we note that another characteristic of the currently proposed platforms is that the endomorphisms suggested for use with \sdpke{} typically involve multiplication by one or more auxiliary matrices; that is,  for a particular semigroup $G_p$, if $(g,\phi)\in G_p\times End(G_p)$ the group element $\phi(g)$ has the form $A\cdot g\cdot B$, where $A,B\in G$ are fixed. If the matrix size is fixed each application of $\phi$ therefore requires some constant number of operations in the underlying ring of the matrix semigroup, which we may assume has size polynomial in $p$. The complexity of this matrix multiplication will be dominated by the multiplication in the underlying ring. Since the size of the underlying ring is also polynomial in $p$, each multiplication has complexity $\bigO{(\log p)^2}$ (since $\bigO{\log poly(p)}=\bigO{\log p}$). We conclude that both multiplication of elements in $G_p$, and evaluation  of $\phi(g)$, can be done in time $\bigO{(\log p)^2}$.

With these observations in mind, we define the following:
\begin{definition}
    Let $P$ some countable indexing set. A family of semigroups $\{G_p\}_{p\in P}$ is said to be \textit{easy} if
    \begin{enumerate}
        \item $|G_p|$ grows monotonically and polynomially in $p$
        \item For any $p$, any tuple $(g,h,\phi)\in G_p\times G_p\times End(G_p)$ is such that $g\cdot h$ and $\phi(g)$ can be evaluated in time $\bigO{(\log p)^2}$.
    \end{enumerate}
\end{definition}

Many of the complexity results within the paper assume that we are dealing with an easy family of semigroups, basically in an attempt to model the behaviour of suggested examples of semigroup family. 

\subsection{Computational Problems}
We have already alluded to some of the hard problems to be found in this paper. Here we give full definitions of all of them to serve as a kind of `glossary' section.

\begin{definition}[Semidirect Discrete Logarithm Problem]
    Given a public (semi)group $G$, its public endomorphism semigroup $End(G)$ and a public pair $(g,\phi)\in G\times End(G)$, let $N$ be the size of the set $\{s(g,\phi,i):i\in\mathbb{N}\}$. Choose $x$ from $\{1,...,N\}$ uniformly at random, calculate $s(g,\phi,x)$ and create the pair $((g,\phi),s(g,\phi,x))$. The Semidirect Discrete Logarithm Problem (\sdlp{}) with respect to $(g,\phi)$ is to recover the integer $x$ given the pair $(g,\phi)$ and $s(g,\phi,x)$.
\end{definition}


We now give the computational problem that we seek to reduce to. This version of the problem is taken from \cite{montgomery2022full}, but can be found as the \textit{vectorisation} problem, respectively, in \cite{couveignes2006hard}.

\begin{definition}[Group Action Discrete Logarithm]
    Given a public commutative group action $(G,X,\star)$, sample $g\in G$ and $x\in X$ uniformly at random, compute $y=g\star x$ and create the pair $(x,y)$. The Group Action Discrete Logarithm Problem (\gadlp{}) with respect to $x$ is to recover $g$ given the pair $(x,y)$.
\end{definition}



\begin{remark}
    The idea of \sdpke{} is used in this paper to motivate the study of \sdlp{}. However, as in the classical case, the security of \sdpke{} is not known to be precisely equivalent to \sdlp{}, and indeed one can define group action-related and semidirect product-related problems in the style of the Computational Diffie-Hellman problem. Studying these problems is beyond the scope of this work.
\end{remark}

Finally we give a seemingly unrelated problem requiring a small amount of introduction. Let $f,g:A\to S$ be injective functions, where $S$ is a set and $A$ is a finite abelian group. We say that $f,g$ \textit{hide} some $s\in A$ if one has $g(a)=f(a+s)$ for each $a\in A$. 

\begin{definition}[Abelian Hidden Shift Problem]\label{hidden-shift}
    Given a public abelian group $A$ and a set $S$, suppose two injective functions $f,g$ hide some $s\in A$. The Abelian Hidden Shift Problem (\ahsp{}) is to recover the group element $s$.
\end{definition}

\section{Structure of the Exponents}\label{group-action-section}

All of the algorithms in this paper rely on the construction of a certain group action - recall that such an object consists of a group, a set, and a function (Section \ref{prelims-group-actions}, Definition \ref{defn-group-actions}). As a general outline to our strategy, we first define and deduce properties of a particular set, from which the appropriate group and function will follow. 

With this in mind, we make the following definition. For now we will dispense with our notion of an easy, parameterised family of semigroups, since the results presented in this section apply to any fixed semigroup. In fact, for compactness of exposition, for the remainder of this section by $G$ we mean an arbitrary finite (semi)group, and by $End(G)$ we mean its associated endomorphism semigroup.
\begin{definition}
For a pair $(g,\phi)\in G\times End(G)$, define
\[\mathcal{X}_{g,\phi}:=\{s(g,\phi,i):i\in\mathbb{N}\}\]
\end{definition}
We will often write $\mathcal{X}_{(g,\phi)}$ as $\mathcal{X}$ when clear from context. Certainly this object is neither a group nor a semigroup - numerous counterexamples can be found whereby multiplication of elements in this set are not contained in the set - but we can make some progress by borrowing from the standard theory of monogenic semigroups; presented, for example, in \cite{howie1995fundamentals}. Since $\mathcal{X}\subset G$, $\mathcal{X}$ is finite --- the set $\{x\in\mathbb{N}:\exists y\neq x\quad s(g,\phi,x)=s(g,\phi,y)\}$ must therefore be non-empty, or the set would be in bijection with the natural numbers, contradicting the fact that $G$ is finite. We may therefore choose the smallest element of this set, say $n$. By definition of $n$ the set $\{x\in\mathbb{N}:s(g,\phi,n)=s(g,\phi,n+x)\}$ must also be non-empty, so we may again pick its smallest element and call it $r$. 

The structure of $\mathcal{X}$ is further restricted by the following result:
\begin{lemma}\label{splitting-property}
Let $(g,\phi)\in G\times End(G)$ and $x,y\in\mathbb{N}$, then 
\[\phi^x\left(s(g,\phi,y)\right)\cdot s(g,\phi,x)=s(g,\phi,x+y)\]
\end{lemma}
\begin{proof}
Note that $s(g,\phi, x+y)=\phi^{x+y-1}(g)\cdot\ldots\cdot g$. Since $\phi$ preserves multiplication, applying $\phi^x$ to $s(g,\phi,y)$ adds $x$ to the exponent of each term. Multiplication on the right by $s(g,\phi,x)$ then completes the remaining terms of $s(g,\phi,x+y)$.
\qed 
\end{proof}

\begin{remark}
One can entirely symmetrically swap the roles of $x$ and $y$ in the above argument, which gives two ways of calculating $s(g,\phi,x+y)$. In essence, therefore, this result gives us a slightly more elegant proof of the correctness of \sdpke{}.
\end{remark}

This method of inducing addition in the integer argument of $s$ is sufficiently important that we will invoke a definition for it.

\begin{definition}\label{ast-function}
    Let $(g,\phi)\in G\times End(G)$ and define a function $f:\mathbb{N}\times \mathcal{X}\to \mathcal{X}$ by 
    \[f(i,s(g,\phi,j))=\phi^i(s(g,\phi,j))\cdot s(g,\phi,i)\]
    where $f(i,s(g,\phi,j))$ may also be written as $i\ast s(g,\phi,j)$.
\end{definition}

\begin{remark}
    Strictly speaking the $\ast$ operation depends on a choice of pair $(g,\phi)$ and so should be written $\ast_{g,\phi}$. However, since the remainder of the paper deals with \sdlp{} when a choice of such a pair is fixed, we will suppress this detail.
\end{remark}

\begin{remark}
    If $G$ is of the type discussed in Section~\ref{prelims-efficiency} -i.e., $G=G_p$ is one of a family of easy semigroups - the value $i\ast s(g,\phi,j)$ can be computed in time $\bigO{(\log i)(\log p)^2}$. provided $s(g,\phi,j)$ is already known. This is because to compute $\phi^i(s(g,\phi,j))$ requires the computation and evaluation of $\phi^i$, the computation of $s(g,\phi,i)$, and some fixed number of multiplications in $G$ - but we know from Section~\ref{prelims-efficiency} that one can calculate $(g,\phi)^i=(s(g,\phi,i),\phi^i)$ in time $\bigO{(\log i)(\log p)^2}$. 
\end{remark}

Thus far we have established that corresponding to any fixed pair $(g,\phi)\in G\times End(G)$ is a set $\mathcal{X}_{g,\phi}=\mathcal{X}$ and a pair of integers $n,r$. By Lemma~\ref{splitting-property} we know that $i\ast s(g,\phi,j)=s(g,\phi,i+j)$ for any $i,j\in\mathbb{N}$, so by definition of $n,r$ we have
\begin{align*}
    s(g,\phi,n+2r) &= r\ast s(g,\phi,n+r) \\
    &= r \ast s(g,\phi,n) \\
    &= s(g,\phi,n+r) = s(g,\phi,n)
\end{align*}

We conclude, by extending this argument in the obvious way, that $s(g,\phi,n+qr)=s(g,\phi,n)$ for each $q\in\mathbb{N}$. In fact, we have the following:
\begin{lemma}\label{r-identity}
Fix $(g,\phi)\in G\times End(G)$ and let $n,r$ be the corresponding integer pair as above. One has that 
\[s(g,\phi,n+x+qr)=s(g,\phi,n+x)\]
for all $x,q\in\mathbb{N}$.
\end{lemma}

We will frequently invoke Lemma~\ref{r-identity}. Indeed, we immediately get that the set $\mathcal{X}$ cannot contain values other than $\{g,...,s(g,\phi,n),...,s(g,\phi,n+r-1)\}$. If any of the values in $\{g,...,s(g,\phi,n-1)\}$ are equal we contradict the minimality of $n$, and if any of the values in $\{s(g,\phi,n),...,s(g,\phi,n+r-1)\}$ are equal we contradict the minimality of $r$. We have shown the following:

\begin{theorem}\label{exponent-structure}
Fix $(g,\phi)\in G\times End(G)$. The set $\mathcal{X}=\{s(g,\phi,i):i\in\mathbb{N}\}$ has size $n+r-1$ for integers $n,r$ dependent on $g,\phi$. In particular
\[\mathcal{X}=\{g,...,s(g,\phi,n),...,s(g,\phi,n+r-1)\}.\]
\end{theorem}

We refer to the set $\{g,...,s(g,\phi,n-1)\}$ as the \textit{tail}, written $\mathcal{T}_{g,\phi}$, of $\mathcal{X}_{g,\phi}$; and the set $\{s(g,\phi,n),...,s(g,\phi,n+r-1)\}$ as the \textit{cycle}, written $\mathcal{C}_{g,\phi}$, of $\mathcal{X}_{g,\phi}$. The values $n_{g,\phi}$ and $r_{g,\phi}$ are called the \textit{index} and \textit{period} of the pair $(g,\phi)$. We shall feel free to omit the subscript at will when clear from context. 

One can see that unique natural numbers correspond to each element in the tail, but infinitely many correspond to each element in the cycle. In fact, each element of the cycle corresponds to a unique residue class modulo $r$, shifted by the index $n$. This is a rather intuitive fact, but owing to its usefulness we will record it formally. In the following we assume the function $\mod$ returns the canonical positive residue.

\begin{theorem}\label{exponent-congruence}
Fix $(g,\phi)\in G\times End(G)$ and let $x,y\in\mathbb{N}$. We have 
\[s(g,\phi,n+x)=s(g,\phi,n+y)\]
if and only if $x\mod r=y\mod r$.
\end{theorem}
\begin{proof}
In the reverse direction, setting $x'=x\mod r$ and $y'=y\mod r$, we have by Lemma~\ref{r-identity} that $s(g,\phi,n+x)=s(g,\phi,n+x')$ and $s(g,\phi,n+y)=s(g,\phi,n+y')$. By assumption $x'=y'$, and $0\leq x',y'<r$. The claim follows since  we know values in the range $\{s(g,\phi,n),...,s(g,\phi,n+r-1)\}$ are distinct by Theorem~\ref{exponent-structure}. 

On the other hand, suppose $s(g,\phi,n+y)=s(g,\phi,n+x)$ but $x\not\equiv y\mod r$. Without loss of generality we can write $y=x'+u+qr$ for some $q\in\mathbb{N}, 0<u<r$ and $x'=x\mod r$. By Lemma~\ref{r-identity}, since $s(g,\phi,n+y)=s(g,\phi,n+x)$ we must have 
\[s(g,\phi,n+x')=s(g,\phi,n+x'+u)\] 
where $s(g,\phi,n+x)=s(g,\phi,n+x')$ also by Lemma~\ref{r-identity}. There are now three cases to consider; we claim each of them gives a contradiction.

First, suppose $x'+u=r$, then $s(g,\phi,n+x')=s(g,\phi,n)$. Since $x'<r$ we contradict minimality of $r$. The case $x'+u<r$ gives a similar contradiction. 

Finally, if $x'+u>r$, without loss of generality we can write $x'+u=r+v$ for some positive integer $v$, so we have $s(M,\phi,n+x')=s(M,\phi,n+v)$. Since $x'\neq v$ (else we contradict $u<r$), and both values are strictly less than $r$, we have a contradiction, since distinct integers of this form give distinct evaluations of $s$. 
\qed 
\end{proof}

\subsection{A Group Action}

It should be clear by now that we are interested in the argument of $s$ in terms of residue classes modulo $r$. Recall that the group of residue classes modulo $r$ is denoted $\ZZr$, and its elements are written as $\modd{i}$. We conclude the section by constructing the action of $\ZZr$ on the cycle $\{s(g,\phi,n),...,s(g,\phi,n+r-1)\}$, where we assume that the operator $\mod r$ returns the unique integer in $\{0,...,r-1\}$ associated to its argument. 

\begin{theorem}\label{free-transitive-action}
Fix $(g,\phi)\in G\times End(G)$ and let $n,r$ be the index and period corresponding to $g,\phi$. Moreover, let $\mathcal{C}$ be the corresponding cycle of size $r$. The abelian group $\ZZ_r$ acts freely and transitively on $\mathcal{C}$.
\end{theorem}
\begin{proof}
First note that Theorem~\ref{exponent-congruence} immediately gives that $j\ast s(g,\phi,i+n)=s(g,\phi,(i+j)\mod r+n)$ for any $j\in\mathbb{N}$. Our current definition of $s$ is not defined for negative integer arguments; nevertheless, we can extend the range of the operator $\ast$ as follows. Let $\ast:\mathbb{Z}\times \mathcal{C}\to \mathcal{C}$ be defined by $$j\ast s(g,\phi,i)=\phi^{j\mod r}(s(g,\phi,i+n))\cdot s(g,\phi,j\mod r)$$ Since $j\mod r\geq0$, as usual we have $j\ast s(g,\phi,i+n)=s(g,\phi,i+j\mod r+n)$; but since $s(g,\phi,i+n)\in\mathcal{C}$, we know $0\leq i < r$, so $i\mod r=i$. It follows that $j\ast s(g,\phi,i+n)=s(g,\phi,(i+j)\mod r+n)$.

In fact, fix some $i\in\mathbb{N}$, and let $\modd{j}$ be a fixed element of $\ZZr$. By definition, every $k\in\modd{j}$ is such that $k\mod r=j'$ for some $j'\in\{0,...,r-1\}$; without loss of generality, $j'=j$. We may therefore define $\oast:\ZZ\times\mathcal{C}\to\mathcal{C}$ by
\[\modd{j}\oast s(g,\phi,i+n)= s(g,\phi,(i+j)\mod r+n)\]
where $j$ is the unique element of $\modd{j}$ such that $k\mod r=j$ for each $k\in\modd{j}$. We claim that $(\ZZr, \mathcal{C},\oast)$ is a free, transitive group action.

First, let us verify that a group action is indeed defined. Certainly $\modd{0}$ fixes every element in $\mathcal{C}$, since $s(g,\phi,(i+0)\mod r+n)=s(g,\phi,i+n)$ for each $i\in\{0,...,r-1\}$. Moreover, one has
\begin{align*}
    \modd{k}\oast(\modd{j}\oast s(g,\phi,i+n)) &= \modd{k}\oast s(g,\phi,(i+j)\mod r+n) \\
    &= s(g,\phi,((i+j)\mod r)+k\mod r+n) \\
    &= s(g,\phi,(i+(j+k))\mod r+n) \\
    &= \modd{j+k}\oast s(g,\phi,i+n)\\
    &=(\modd{k}+\modd{j})\oast s(g,\phi,i+n)
\end{align*}

It remains to check that the action is free and transitive. If $\modd{j}\in\ZZr$ is such that $\modd{j}$ fixes an arbitrary element of $\mathcal{C}$, say $s(g,\phi,i+n)$, then we have $s(g,\phi,(i+j)\mod r+n)=s(g,\phi,i+n)$. By Theorem~\ref{exponent-congruence}, we must have $i+j\equiv i\mod r$, so $\modd{j}=\modd{0}$ and the action is free. Moreover, for arbitrary $s(g,\phi,i+n),s(g,\phi,j+n)\in\mathcal{C}$, $\modd{k}=\modd{j-i}\in\ZZr$ is such that $\modd{k}\oast s(g,\phi,i+n)=s(g,\phi,j+n)$, so the action is also transitive and we are done.

\qed 
\end{proof}

We summarise the above by noting that for each $(g,\phi)\in G\times End(G)$ we have shown the existence of a free, transitive, commutative group action $(\mathbb{Z}_r,\mathcal{C},\oast)$, where $r$ and $\mathcal{C}$ depend on the choice of pair $(g,\phi)$.

\section{Group Action Discrete Logarithms}\label{group-action-discrete-log}
Now that we have established the group action, we recall the Group Action Discrete Logarithm Problem (\gadlp{}) from the introduction. Roughly speaking, for a free transitive group action $(G,X,\star)$, and $x,y$ sampled uniformly at random from the set $X$, we are tasked with recovering the unique $G$-element $g$ such that $g\star x=y$. In this section we will show that one can construct a quantum reduction from $\sdlp{}$ to $\gadlp{}$. 

More precisely, we target the type of structure discussed in Section~\ref{prelims-efficiency}; that, is a set of finite semigroups $\{G_p\}_p$ indexed by some parameter $p$, such that the size of each $G_p$ is polynomial in $p$ - the so-called `easy' families of semigroups. We know that multiplication in each $G_p$ requires a number of operations bounded above by some constant independent of $p$, and that the complexity of these operations is bounded above by $\bigO{(\log p)^2}$ 

With all this in mind let $\{G_p\}_p$ be such a family of semigroups. In the previous section we have shown that for a fixed $p$, to each pair $(g,\phi)\in G_p\times End(G_p)$ is associated a pair $(n,r)$ and a set $\mathcal{C}$. In this section we seek to show there is an efficient quantum algorithm to solve \sdlp{} with respect to an arbitrary choice of $(g,\phi)$, provided one has access to a \gadlp{} oracle for the group action $(\ZZr, \mathcal{C}, \oast)$. 

Before giving this reduction there remains a significant obstacle to overcome: for an arbitrary pair $(g,\phi)$ we have only proved the existence of the corresponding values $n,r$, but we do not have a means of calculating them. In order to provide a reduction to a \gadlp{} oracle, however, we need to specify the appropriate group action. We therefore require access to the values $n,r$ - in the next section, we will provide a quantum method of recovering these integers. We note that assuming access to a quantum computer is, for our purposes, justified since the best-known algorithms for \gadlp{} are quantum anyway.

\subsection{Calculating the Index and Period}\label{non-empty-tails}

In order to reason on the complexity of our algorithm we will use the following worst-case indicator, defined as follows:
\begin{definition}\label{extremal-functions}
    Let $\{G_p\}_{p\in P}$ be an easy family of finite semigroups parameterised by some set $P$. Define the following function on $P$:
    \[N(p)=\max_{(g,\phi)\in G_p\times End(G_p) }|\mathcal{T}_{g,\phi}+\mathcal{C}_{g,\phi}|\]
\end{definition}

The function $N(p)$ gives a bound on the size of $\mathcal{X}_{g,\phi}$ for any $(g,\phi)\in G_p\times End(G_p)$. Since a crude such bound is the size of an easy semigroup $G_p$, which is assumed polynomial in $p$, we have that $N(p)$ is at worst polynomial in $p$.

Our method of calculating the index and period borrows heavily from ideas in \cite[Theorem~1]{childs2014quantum}, which is itself a slightly repurposed version of \cite[Algorithm~5]{childs2010quantum}. Indeed, after a certain point we will be able to quote methods of these algorithms verbatim - nevertheless, to cater to our specific context it remains incumbent upon us to justify the following.

\begin{lemma}\label{exp-calc}
    Let $\{G_p\}_p$ be an easy family of semigroups, and for an arbitrary $p$ fix a pair $(g,\phi)\in G_p\times End(G_p)$. For any $l\in\mathbb{N}$, one can construct the superposition
    \[\frac{1}{\sqrt{M}}\sum_{k=0}^{M-1}\ket{k}\ket{s(g,\phi,k)}\]
    in time $\bigO{(\log M) (\log p)^2}$, where $M=2^l$.
\end{lemma}
\begin{proof}
    When $(g,\phi)$ is fixed, notice that we can think of $s(g,\phi,i):G\times End(G)\times\NN\to G$ as a function $s_{g,\phi}(i):\NN\to G$. Since $N(p)$ is a bound on the size of $\mathcal{X}_{g,\phi}$, taking $m$ to be smallest integer such that $2^m\geq N(p)$ (note that $m=\bigO{\log(N(p))}$), the set $\mathcal{X}_{g,\phi}$ has binary representation in the set $\{0,1\}^m$. By definition the integers $\{0,...,M-1\}$ have binary representation in $\{0,1\}^l$, so we can think of the restriction of $s_{g,\phi}$ on $\{0,...,M-1\}$ as a function from $\{0,1\}^{l}$ into $\{0,1\}^m$. There is therefore a Boolean circuit computing $s_{g,\phi}$; it is standard (say, by \cite[Theorem~2.3.2]{hirvensalo2003quantum}) that one can construct a circuit implementing $s_{g,\phi}$ using reversible gates. Call this circuit $Q_{s_{g,\phi}}$; since the reversible circuit does no worse than the classical circuit, we can assume a single application of $Q_{s_{g,\phi}}$ takes at worst the time complexity of calculating $s_{g,\phi}(M)$. 

    What is the time complexity of calculating $s_{g,\phi}(M)$? We know by definition that $(s_{g,\phi}(M),\phi^M)=(g,\phi)^M$, where the exponentiation refers to holomorph exponentiation. Recall that since we have assumed we are working in an easy family of semigroups, each multiplication in the holomorph involves one application of $\phi$, followed by some fixed constant number of matrix multiplications independently of $p$. Calculating the holomorph exponentiation in the standard square-and-multiply fashion, therefore, we expect to perform $\bigO{\log M}$ holomorph multiplications. We know that evaluating $\phi$ takes time $\bigO{(\log p)^2}$; since a fixed number of matrix multiplications follow, the total time for calculating $s(g,\phi,M)$ is $\bigO{\log M(\log p)^2}$.
     
    If we can show a single application of $Q_{s_{g,\phi}}$ gives the desired superposition we are done. It is standard, however, that the uniform superposition of an $M$-bit quantum register, together with an ancillary $m$-bit register in the state $\ket{0}$, can be inputted into $Q_{s_{g,\phi}}$ to produce the desired superposition. Since preparing the appropriate uniform superposition can be done by applying a Hadamard gate in time $\bigO{\log M}$, we are done.
    \qed
\end{proof}

Armed with the ability to efficiently calculate the appropriate superposition, we will quickly find ourselves with exactly the kind of state arrived at in \cite[Algorithm~5]{childs2010quantum}, thereby allowing us to recover the period $r$ in Algorithm~\ref{period-recovery}. A small adaptation of standard binary search techniques completes the task by using knowledge of $r$ to recover the index $n$.

\begin{algorithm}[h!]
    \caption{PeriodRecovery($((g,\phi),M)$)}\label{period-recovery}
    {\bf Input}: Pair $(g,\phi)\in G_p\times End(G_p)$, upper bound on size of superposition to create $M$ 
    \par \noindent
    {\bf Output}: Period $r$ of $(g,\phi)$ or 'Fail'
    \par \medskip \noindent
    \begin{algorithmic}[1]
        \State $R_0\gets\ket{0}\ket{0}$\label{init} 
        \State $R_1\gets$ Hadamard transform applied to first register
        \State $R_2\gets$ appropriate quantum circuit applied to $R_1$\label{comp-sup}
        \State Measure second register leaving collapsed first register $R_3$\label{obs:R}
        \State $R_4\gets$ QFT over $\mathbb{\ZZ_M}$ applied to $R_3$\label{qft}
        \State $R_5\gets$ measure $R_4$\label{obs:R4}
        \State $r\gets$ continued fraction expansion of $R_5/M$\label{pp} 
        \If{$r\ast s(g,\phi,M)\neq s(g,\phi,M)$}
            \State \Return `Fail'
        \Else
            \State \Return $r$
        \EndIf
    \end{algorithmic}
\end{algorithm}

\begin{algorithm}[h!]
    \caption{BinarySearch($(g,\phi),start,end,r$)}\label{index-recovery}
    {\bf Input}: Pair $(g,\phi)$, integers $start,end$ where $start\leq end$, period $r$ of $g,\phi$
    \par \noindent
    {\bf Output}: Index $n$ of $(g,\phi)$
    \par \medskip \noindent
    \begin{algorithmic}[1]
        \If{$start=end$}:
            \State \Return $start$
        \EndIf
        \State $left\gets start$
        \State $right\gets end$
        \State $mid\gets\lfloor(left+right)/2\rfloor$
        \If{$r\ast s(g,\phi,mid)\neq s(g,\phi,mid)$}
            \State \Return BinarySearch($(g,\phi),mid+1,right,r$)
        \Else
            \State \Return BinarySearch($(g,\phi),left,mid,r$)
        \EndIf
    \end{algorithmic}
\end{algorithm}

\begin{theorem}\label{index-period-recovery}
    Let $\{G_p\}_p$ be an easy family of semigroups, and fix $p$. For any pair $(g,\phi)\in G_p\times End(G_p)$:
    \begin{enumerate}
        \item For sufficiently large $M\in\NN$, PeriodRecovery($(g,\phi),M$) recovers the period $r$ of $(g,\phi)$ in time $\bigO{(\log p)^3}$, and with constant probability.
        \item BinarySearch($(g,\phi),1,M,r$) returns the index $n$ of $g,\phi$ in time $\bigO{(\log p)^4}$.
    \end{enumerate}   
\end{theorem}
\begin{proof}
\begin{enumerate}
    \item Fix a pair $(g,\phi)\in G_p\times End(G_p)$ and let $r$ be its period. Let $\ell\in\mathbb{N}$ be the smallest positive integer such that $2^{\ell}\geq(N(p)^2+N(p))$, and $M=2^\ell$. In steps~\ref{init}-\ref{comp-sup} of Algorithm~\ref{period-recovery}, we prepare the required superposition as described in Lemma~\ref{exp-calc}.
    
    In Step \ref{obs:R}, we measure the second register. With probability $n/M$ doing so will cause us to observe an element of the tail; that is, some $s(g,\phi,i)$ such that $i < n$. In this case, by the laws of partial observation, the first register is left in a superposition of integers corresponding to this value - but by definition there is only one of these, so the first register consists of a single computational basis state and the algorithm has failed. On the other hand, with probability $(M-n)/M$ measuring the second register gives an element of $\mathcal{C}$. Now, since $M\geq N(p)^2+N(p)$, we observe an element of $\mathcal{C}$ with probability
    \[\frac{M-n}{M}=1-\frac{n}{M}\geq 1-\frac{n}{N(p)^2+N(p)}\]
    Since by definition one has $n\leq N(p)$, it follows that the relevant probability is better than $N(p)/(N(p)+1)\geq 1/2$. In other words, we observe an element of the desired form with constant, positive probability.
    Provided such an element was observed, after measuring the second register, the superposition of corresponding integers in the first register is the following:
    \[\frac{1}{\sqrt{s_r}}\sum_{j=0}^{s_r-1}\ket{x_0+jr}\]
    To see this, note that the function $s$ is periodic of period $r$, and by Theorem~\ref{exponent-structure} each $s(g,\phi,i)$ such that $i\geq n$ can only assume one of the distinct values $s(g,\phi,n),...,s(g,\phi,n+r-1)$. In particular, the integers in $\{1,...,M\}$ that give a specific value of the cycle under $s$ are of the form $x_0+jr$ for some $x_0\in\{n,...,n+r-1\}$. The largest such integer, by definition, is $x_0+s_rr$, where $s_r$ is just the largest integer such that $x_0+s_rr<M$. Note that the superposition is normalised by this factor so that the sum of the squares of the amplitudes is $1$.
    
    We now have exactly the same kind of state found in \cite[Algorithm~5]{childs2010quantum}\footnote{This type of state also occurs in Shor's factoring algorithm.}, so we may proceed exactly according to the remaining steps in this algorithm. In Step~\ref{qft} we apply a Quantum Fourier Transform (QFT) over $\mathbb{Z}_M$ to the state, which can be done in time $\mathcal{O}((\log M)^2)$. In step~\ref{obs:R4} we measure the state $R_4$; it is shown in \cite[Algorithm~5, Step~5]{childs2010quantum} that with probability at least $4/\pi^2$, measuring the resulting state leaves one with the closest integer to one of the at most $r$ multiples of $M/r$ (note that $M/r$ is not necessarily an integer) with probability better than $4/\pi^2$. Writing this closest integer as $\lfloor jM/r\rceil$ for some $j\in\mathbb{N}$, one checks that the fraction $j/r$ is a distance of at most $1/2M$ from $(\lfloor jM/r \rceil)/M$; by \cite[Theorem~8.4.3]{hirvensalo2003quantum}, $j/r$ will appear as one of the convergents in the continued fraction expansion of $(\lfloor jM/r \rceil)/M$ provided $1/2M\leq1/2r^2$. Certainly this holds, since $r<N(p)<M$. Provided we have observed an integer of the appropriate form, then, it remains to carry out a continued fraction expansion on $(\lfloor jM/r \rceil)/M$, which we can do with repeated application of the Euclidean algorithm.

    Let us summarise the complexity of the algorithm. The dominating factors are the creating of the relevant superposition in time 
    \[\bigO{(\log M)(\log p)^2}=\bigO{(\log N(p)(\log p)^2=\bigO{(\log p)^3}}\]
    where the last equality follows from the easy property of the relevant semigroup family; that is, one has that $N(p)$ is at worst polynomial in $p$. Similarly, the application of QFT can be done in time $\bigO{(\log p)^2}$, so we have the complexity estimate claimed at the outset. Note also that the algorithm succeeds provided an element of the cycle is observed after the first measurement, and that the second measurement gives an appropriate integer. Since both of these events occur with probability bounded below by a constant, the algorithm succeeds with probability $\bigOmega{1}$.

    \item We prove correctness of the algorithm by proving that any values $start,end$ such that $start\leq n\leq end$ will return $n$, which we accomplish by strong induction on $k=start-end+1$. To save on cumbersome notation we assume $(g,\phi)$ and $r$ are fixed, and write
    \[BinarySearch((g,\phi,),start,end,r)=BS(start,end)\]

    First, suppose $k=1$ and $start\leq n\leq end$. Either $n=start$ or $n=start+1$, and we know that $mid=start$ after the floor function is applied. In the first case, $r\ast s(g,\phi,mid)=s(g,\phi,mid)$, so BS($start,mid$) is returned; but since $start=mid$, $start=n$ is returned. Otherwise, one has $r\ast s(g,\phi,mid)\neq s(g,\phi,mid)$ and BS($mid+1,end$) is returned, and we are done since $mid+1=end=n$.

    Now for some $k>1$ suppose all positive integers $start', end'$ such that $start'\leq n\leq end'$ and $end'-start'+1<k$ have BS($start',end'$)=$n$. We should like to show that an arbitrary choice of $start,end$ with $start\leq n\leq end$ and $end-start+1=k$ enjoys this same property. To see that it does we can again consider the two cases. 

    The algorithm first calculates $mid=\lfloor(end-start)/2\rfloor$. Suppose $r\ast s(g,\phi,mid)=s(g,\phi,mid)$, then BS($start,mid$) is run. Since $n$ is the smallest integer such that $r\ast s(g,\phi,n)=s(g,\phi,n)$ and $n\geq start$ by assumption, we know $start\leq n\leq mid$. Moreover, $mid-start+1<end-start+1<k$. By inductive hypothesis BS($start,mid$) returns $n$.

    The other case is similar; this time, if $r\ast s(g,\phi,mid)\neq s(g,\phi,mid)$ we know $n\geq mid+1$ by definition of $n$. We also know that $end-(mid+1)+1=end-mid<end-mid+1=k$, so the algorithm returns BS($mid+1,end$)=$n$ by inductive hypothesis.

    Notice that each time $BinarySearch$ is called the calculation of $r\ast s(g,\phi,mid)$ is required. We know already that $s(g,\phi,mid)$ can be calculated in time $\bigO{\log mid(\log p)^2}=\bigO{(\log p)^3}$. Given $\phi^{r}$, $s(g,\phi,r_{g,\phi)}$ and $s(g,\phi,mid)$, the calculation of $r\ast s(g,\phi,mid)$ requires evaluating an endomorphism and a semigroup multiplication - we have argued already that this can be done in time $\bigO{(\log p)^2}$. Recall that in the proof of Lemma~\ref{exp-calc}, we showed that one can calculate the holomorph exponent $(g,\phi)^{r}=(s(g,\phi,r),\phi^{r})$ in time $\bigO{\log r(\log p)^2}$, so the total calculation is done in time $\bigO{(\log p)^3}$ since $r<M$. Clearly, BinarySearch will be called $\bigO{\log M}=\bigO{\log p}$ times, since the size of the interval to search halves at each iteration, and we conclude that $BinarySearch$ recovers the index in time $\bigO{(\log p)^4}$. 
\end{enumerate}
\qed
\end{proof}

\subsection{From \sdlp{} to \gadlp{}}
Let us assemble the components developed so far in this section into a reduction of $\sdlp{}$ to $\gadlp{}$.

\begin{theorem}\label{sdlp-to-gadlp}
Let $\{G_p\}_p$ be an easy family of semigroups, and fix $p$. Algorithm~\ref{sd-ga} solves \sdlp{} with respect to a pair $(g,\phi)\in G_p\times End(G_p)$ given access to a \gadlp{} oracle for the group action $(\ZZ_{r}, \mathcal{C}, \oast)$. The algorithm runs in time $\bigO{(\log p)^4}$, makes at most a single query to the \gadlp{} oracle, and succeeds with probability $\bigOmega{1}$.  
\end{theorem}

\begin{algorithm}
    \caption{Solving \sdlp{} with \gadlp{} oracle}\label{sd-ga}
    {\bf Input}: $(g,\phi), s(g,\phi,x)$
    \par \noindent
    {\bf Output}: $x$
    \par \medskip \noindent
    \begin{algorithmic}[1]
        \State $r\gets$PeriodRecovery$((g,\phi), M)$ for sufficiently large $M$
        \If{$r$='Fail'}
        \State\Return `Fail'
        \EndIf
        \State $n\gets$BinarySearch($(g,\phi),1,M,r$)
        \If{$r\ast s(g,\phi,x)=s(g,\phi,x)$}\label{check}
            \State $d\gets s(g,\phi,n)$\label{cycle-point}
            \State $x'\gets$ \gadlp{} oracle applied to $d,s(g,\phi,x)$\label{oracle-query}
            \State $x\gets n+x'$\label{x-recovered-1}
        \Else
            \State $t\gets$BinarySearch2($s(g,\phi,x),1,n,r$)
            \State $x\gets n-t$
        \EndIf
        \State \Return $x$
    \end{algorithmic}
\end{algorithm}
\begin{proof}
Consider an instance of \sdlp{} whereby we are given the pair $(g,\phi)$ and the value $s(g,\phi,x)$, for some $x$ sampled uniformly at random from the set $\{1,..,n+r-1\}$. We show that Algorithm~\ref{sd-ga} recovers $x$.

We start by applying Algorithms 1 and 2 to the pair $(g,\phi)$, recovering the pair $n,r$ with constant probability. By Theorem~\ref{index-period-recovery}, we can do so in time $\bigO{(\log p)^4}$. Now, $s(g,\phi,x)$ might be in tail or in the cycle - but with our knowledge of $r$ we can check in Step~\ref{check} which is true by verifying whether $r\ast s(g,\phi,x)=s(g,\phi,x)$. As discussed in the proof of Theorem~\ref{index-period-recovery}, we can perform this check in time $\bigO{(\log p)^3}$. 

There are now two cases to consider. First, suppose that the check in Step~\ref{check} is passed, then $s(g,\phi,x)$ is in the cycle, and we may proceed as follows. Compute $s(g,\phi,n)$ in time $\bigO{(\log p)^3}$, and query the \gadlp{} oracle on input $s(g,\phi,n),s(g,\phi,x)$ (Step~\ref{oracle-query}) to recover the $\ZZr$ element $\modd{y}$. Without loss of generality the smallest positive representative of this class, say $x'$, is such that $n+x'=x$, so we recover $x$ in Step~\ref{x-recovered-1}.

Now suppose that $s(g,\phi,x)$ is in the tail. We run the algorithm BinarySearch2 to recover $t$, the smallest integer such that $t\ast s(g,\phi,x)$ is invariant under $r$. BinarySearch2 is precisely the same as Algorithm~\ref{index-recovery}, except that in the verification step, we check if $r\ast(mid\ast s(g,\phi,x))=mid\ast s(g,\phi,x)$. It is not hard to adapt the proof of correctness to show that BinarySearch2 does indeed return $t$ in time $\bigO{(\log p)^4}$. Moreover, by minimality of $n$ and the additivity of $\ast$, we must have $x+t=n$, from which we recover $x=n-t$.

Finally, we note that the only probablistic step of this algorithm is the application of Algorithm~\ref{period-recovery}, so we successfully recover $x$ with the same success probability as Algorithm~\ref{period-recovery}, and we are done.
\qed 
\end{proof}

In summary, we have an efficient quantum reduction from \sdlp{} to \gadlp{}: an efficient quantum procedure extracts the period $r$, and from there a classical procedure gives the index $n$. In order to recover $x$, it remains to either carry out an efficient classical procedure, or recover $x$ with a single query to a \gadlp{} oracle. Moreover, assuming the \gadlp{} oracle always succeeds, the success probability is precisely that of Algorithm 1 - that is, bounded below by a positive constant independently of $p$. 

\begin{remark}
    The factor $\log p$ in the complexity estimate is really coming from the `length' of a binary representation of $G_p$; that is, the number of bits required to represent $G_p$. In our case the size of $G_p$ happens to be polynomial in $p$, and therefore the relevant `length' is of order $\bigO{\log p}$. One might be used to seeing the complexity of similar period-finding routines, such as Shor's algorithm, presented as cubic in the length of a binary representation of the relevant parameters - see for example \cite[Section~3.3.3]{hirvensalo2003quantum}. In our case, the total complexity is quartic in the length of a binary representation, essentially because after the quantum part of the algorithm we still need to compute $\bigO{\log p}$ evaluations of the function $s$ in order to compute the index. In a sense, then, we can think of this extra $\log p$ factor as the extra cost incurred from the slightly more complicated scenario inherent to the problem.
\end{remark}

\section{Quantum Algorithms for \gadlp{}}\label{gadlp-algos}
Now that we have shown \sdlp{} can be efficiently solved with access to an appropriate \gadlp{} oracle it remains to examine the state of the art for \gadlp{}. It is here that the Abelian Hidden Shift Problem (Definition~\ref{hidden-shift}) comes in to play. Roughly speaking, we are given two injective functions $f,g$ from a group $A$ to a set $S$ that differ by a constant `shift' value, and our task is to recover the shift value. 

It is reasonably well-known (see \cite{stolbunov2010constructing, childs2014constructing}) that \gadlp{} reduces to \ahsp{}. In this section, we provide a context-specific proof of this fact, before discussing the best known algorithms for \ahsp{}.

\subsection{Group Actions to Hidden Shift}
The following result is found more or less verbatim in, for example, \cite{childs2014constructing}. We here give a context-specific reduction, for completeness.

\begin{theorem}\label{gadlp-to-ahsp}
    Let $\{G_p\}_p$ be an easy family of semigroups and fix $p$. For some pair $(g,\phi)\in G_p\times End(G_p)$ let $(\ZZr, \mathcal{C}, \oast)$ be the associated group action defined in Theorem~\ref{free-transitive-action}. One can efficiently solve \gadlp{} in $(\ZZr, \mathcal{C}, \oast)$ given access to an \ahsp{} oracle with respect to $\ZZr, \mathcal{C}$.
\end{theorem}
\begin{proof}
    Suppose we are given an instance of \gadlp{} in $(\ZZr, \mathcal{C}, \oast)$; that is, we are given a pair $(s(g,\phi,n+i),s(g,\phi,n+j))\in\mathcal{C}$ for some $i,j\in\{1,...,r\}$ and tasked with finding the unique $\modd{k}\in\ZZr$ such that $\modd{k}\oast s(g,\phi,n+i)=s(g,\phi,n+j)$. Our strategy is to construct injective functions $f_A,f_B:\ZZr\to\mathcal{C}$ that hide $\modd{k}$, and use the \ahsp{} oracle to recover this value. 

    Set $f_A,f_B:\ZZr\to\mathcal{C}$ as $f_A(\modd{x})=\modd{x}\ast s(g,\phi,n+i)$ and $f_B(\modd{x})=\modd{x}\ast s(g,\phi,n+j)$. Then
    \begin{align*}
        f_B(\modd{x})&=\modd{x}\ast s(g,\phi,n+j) \\
        &= \modd{x} \ast (\modd{k}\ast s(g,\phi,n+i)) \\
        &= (\modd{x}+\modd{k})\ast s(g,\phi,n+i) \\
        &= f_A(\modd{x}+\modd{k})
    \end{align*}
    In other words, $f_A,f_B$ hide $\modd{k}$. To complete the setup of an instance of \ahsp{} we require the functions to be injective, which follows from the action being free and transitive.
    \qed
\end{proof}

Note that we have in this case left out complexity estimates. This is because in order to give a full description of the functions $f_A,f_B$ we need to compute the group $\ZZr$, which can be done efficiently with knowledge of $r$. However, since we have already described a method of recovering $r$, we will discuss the complexity in the full \sdlp{} algorithm at the end of this section.

\subsection{Hidden Shift Algorithms}
We have finally arrived at the problem for which there are known quantum algorithms. The fastest known is of subexponential complexity, and is presented in \cite[Proposition~6.1]{kuperberg2005subexponential} as a special case of the Dihedral Hidden Subgroup Problem. 

\begin{theorem}[Kuperberg's Algorithm]
    There is a quantum algorithm that solves \ahsp{} with respect to $\ZZr,\mathcal{C}$ with time and query complexity $2^{\bigO{\sqrt{\log r}}}$.
\end{theorem}

Kuperberg's algorithm also requires quantum space $2^{\bigO{\log r}}$. For a slower but less space-expensive algorithm, we can also use a generalised version of an algorithm due to Regev \cite{regev2004subexponential}. The generalised version appears in \cite[Theorem~5.2]{childs2014constructing}. 

\begin{theorem}[Regev's Algorithm]
      There is a quantum algorithm that solves \ahsp{} with respect to $\ZZr,\mathcal{C}$ with time and query complexity 
      \[e^{\sqrt{2}+o(1)\sqrt{\ln r\ln\ln r}}\]
      and space complexity $\bigO{poly(\log r)}$.
\end{theorem}

We note that both Kuperberg's and Regev's algorithms succeed with constant probability.

\subsection{Solving \sdlp{}}
We finish the section by stitching all the components together into an algorithm that solves \sdlp{}. For brevity of exposition we include only complexity estimates for using Kuperberg's algorithm - but finding the bounds in the case of Regev's algorithm is very similar. 
\begin{theorem}\label{complete-algo}
    Let $\{G_p\}_p$ be an easy family of semigroups, and fix $p$. For any pair $(g,\phi)\in G_p\times End(G_p)$, there is a quantum algorithm solving \sdlp{} with respect to $(g,\phi)$ with time and query complexity $2^{\bigO{\sqrt{\log p}}}$. 
\end{theorem}
\begin{proof}
Let $(g,\phi)\in G_p\times End(G_p)$ and suppose we are given the value $s(g,\phi,x)$ for some $x$ sampled uniformly from the set $\{1,...,N\}$, where $N$ is the size of $\mathcal{X}_{g,\phi}$. The following steps recover $x$:
\begin{enumerate}
    \item Run Algorithms~\ref{period-recovery} and~\ref{index-recovery} on the pair $(g,\phi)$. By Theorem~\ref{index-period-recovery}, with positive probability we recover the index and period of $(g,\phi)$, the pair $(n,r)$, in time $\bigO{(\log p)^4}$.
    \item By Theorem~\ref{sdlp-to-gadlp}, either we are done efficiently, or it remains to solve an instance of \gadlp{} with respect to the group action $(\ZZr, \mathcal{C}, \oast)$.
    \item By Theorem~\ref{gadlp-to-ahsp}, once we have computed the group action $(\ZZr,\mathcal{C},\oast)$ it remains to solve an instance of \ahsp{} with respect to $\ZZr,\mathcal{C}$. This can be done with access to the index and period $n,r$.
    \item Solve \ahsp{} using Kuperberg's algorithm or Regev's algorithm.
\end{enumerate}
In summary, the total quantum complexity of solving an \sdlp{} instance for any pair in $G_p\times End(G_p)$ is either $\bigO{(\log p)^4}$, or if a call to the \gadlp{} oracle is required, $2^{\bigO{\sqrt{\log r}}}=2^{\bigO{\sqrt{\log p}}}$ since $G_p$ is from an easy family of semigroups. Depending on constants, we expect this latter term to dominate the complexity. Moreover, we note that since both our algorithm to extract the period and Kuperberg's algorithm succeed with constant probability, we expect our algorithm to succeed with constant probability also. \qed
\end{proof}
\section{Conclusion}
We have provided the first dedicated analysis of \sdlp{}, showing a reduction to a well-studied problem. Perhaps the most surprising aspect of the work is the progress made by a simple rephrasing; we made quite significant progress through rather elementary methods, and we suspect much more can be made within this framework. 

The reader may notice that we have shown that \sdpke{} shares a very similar structure to that of a commutative action-based key exchange; it is known that breaking all such protocols can be reduced to the Abelian Hidden Shift Problem. Indeed, this work shows the algebraic machinery of \sdpke{} is a step twoards a candidate for what Couveignes calls a \textit{hard homogenous space}\footnote{Another major example of which arises from the theory of isogenies between elliptic curves - see, for example, \cite{castryck2018csidh}} \cite{couveignes2006hard}, which was not known until now. In line with the naming conventions in this area we propose a renaming of \sdpke{} to \spdh{}, which stands for `Semidirect Product Diffie--Hellman', and should be pronounced \textit{spud}. 

We would also like to stress the following sentiment. The purpose of this paper is not to claim a general purpose break of \sdpke{} (or, indeed, \spdh{}) - the algorithm presented is subexponential in complexity, which has been treated as tolerable in classical contexts. Instead, the point is to show a connection between \sdlp{} and a known hardness problem, thereby providing insight on a problem about which little was known.

 \subsection*{Acknowledgements}
 We thank Chloe Martindale of the University of Bristol and Samuel Jaques of the University of Oxford for fruitful discussion on this topic at the recent ICMS conference \textit{Foundations and Applications of Lattice-based Cryptography}. We furthermore thank the organisers of this conference for facilitating this conversation. DK thanks Institut des Hautes \'Etudes Scientifiques - IHES for providing stimulating environment while this project was partially done. Authors CB, DK, and LP conducted this work partially with the support of ONR Grant 62909-24-1-2002.

\printbibliography

\end{document}